\PassOptionsToPackage{unicode}{hyperref}
\PassOptionsToPackage{hyphens}{url}
\documentclass[
]{article}
\usepackage{amsmath,amssymb}
\usepackage{iftex}
\ifPDFTeX
  \usepackage[T1]{fontenc}
  \usepackage[utf8]{inputenc}
  \usepackage{textcomp} % provide euro and other symbols
\else % if luatex or xetex
  \usepackage{unicode-math} % this also loads fontspec
  \defaultfontfeatures{Scale=MatchLowercase}
  \defaultfontfeatures[\rmfamily]{Ligatures=TeX,Scale=1}
\fi
\usepackage{lmodern}
\ifPDFTeX\else
\fi
\IfFileExists{upquote.sty}{\usepackage{upquote}}{}
\IfFileExists{microtype.sty}{% use microtype if available
  \usepackage[]{microtype}
  \UseMicrotypeSet[protrusion]{basicmath} % disable protrusion for tt fonts
}{}
\makeatletter
\@ifundefined{KOMAClassName}{% if non-KOMA class
  \IfFileExists{parskip.sty}{%
    \usepackage{parskip}
  }{% else
    \setlength{\parindent}{0pt}
    \setlength{\parskip}{6pt plus 2pt minus 1pt}}
}{% if KOMA class
  \KOMAoptions{parskip=half}}
\makeatother
\usepackage{xcolor}
\usepackage[margin=1in]{geometry}
\usepackage{longtable,booktabs,array}
\usepackage{calc} % for calculating minipage widths
\usepackage{etoolbox}
\makeatletter
\patchcmd\longtable{\par}{\if@noskipsec\mbox{}\fi\par}{}{}
\makeatother
\IfFileExists{footnotehyper.sty}{\usepackage{footnotehyper}}{\usepackage{footnote}}
\makesavenoteenv{longtable}
\usepackage{graphicx}
\makeatletter
\def\maxwidth{\ifdim\Gin@nat@width>\linewidth\linewidth\else\Gin@nat@width\fi}
\def\maxheight{\ifdim\Gin@nat@height>\textheight\textheight\else\Gin@nat@height\fi}
\makeatother
\setkeys{Gin}{width=\maxwidth,height=\maxheight,keepaspectratio}
\makeatletter
\def\fps@figure{htbp}
\makeatother
\usepackage{amsmath}
\usepackage{amssymb}
\usepackage{amsthm}
\usepackage{newunicodechar}
\newunicodechar{ρ}{\ensuremath{\rho}}
\newunicodechar{τ}{\ensuremath{\tau}}
\newunicodechar{θ}{\ensuremath{\theta}}
\newtheorem{definition}{Definition}
\newtheorem{proposition}{Proposition}
\newtheorem{corollary}{Corollary}
\newtheorem{remark}{Remark}
\newtheorem{example}{Example}

\usepackage{booktabs}
\usepackage{longtable}
\usepackage{array}
\usepackage{multirow}
\usepackage{wrapfig}
\usepackage{float}
\usepackage{colortbl}
\usepackage{pdflscape}
\usepackage{tabu}
\usepackage{threeparttable}
\usepackage{threeparttablex}
\usepackage[normalem]{ulem}
\usepackage{makecell}
\usepackage{xcolor}
\ifLuaTeX
  \usepackage{selnolig}  % disable illegal ligatures
\fi
\usepackage[]{natbib}
\usepackage{bookmark}
\IfFileExists{xurl.sty}{\usepackage{xurl}}{} % add URL line breaks if available
\hypersetup{
  pdftitle={Covariate-adjusted statistical dependence representation through partial copulas: bounds and new insights},
  pdfauthor={Vinícius Litvinoff Justus (Carlos III University of Madrid, Spain); Felipe Fontana Vieira (Ghent University, Belgium)},
  pdfkeywords={Copula theory; Rosenblatt transformation; partial
correlation; causal inference},
  hidelinks,
  pdfcreator={LaTeX via pandoc}}

\title{Covariate-adjusted statistical dependence representation through
partial copulas: bounds and new insights}
\author{Vinícius Litvinoff Justus (Carlos III University of Madrid,
Spain)\footnote{\href{mailto:viniliff@gmail.com}{\nolinkurl{viniliff@gmail.com}};
  \href{mailto:vlitvino@est-econ.uc3m.es}{\nolinkurl{vlitvino@est-econ.uc3m.es}}} \and Felipe
Fontana Vieira (Ghent University, Belgium)}
\date{May 2026}

\begin{document}
\maketitle
\begin{abstract}
In this paper, we revisit the notion of partial copula, originally
introduced to test conditional independence, highlighting its capability
to represent the dependence between two random variables after removing
their dependence with a covariate. Building upon results previously
presented in the literature, we show that partial copulas can be seen as
a nonlinear analogue of partial correlation. Then, we prove several
results showing how dependence properties of the conditional copulas
constrain the form of the partial copula. Finally, a simulation study is
conducted to illustrate the results and to show the potential of the
partial copula as a way to describe covariate-adjusted statistical
dependence. This highlights the potential of the method to be used in
causal inference problems and to recover the true sign of a causal
effect.
\end{abstract}

\section{Introduction}\label{introduction}

Quantifying the strength of the statistical association between two
random variables, controlled by one or more variables (usually called
covariates or control variables), is an old problem in statistics and
widely discussed in the context of linear models and regression
coefficients \citep{kutner2005, baba2004, kim2015}. It is well
established in the statistical literature that, under wide conditions,
conditionally independent random variables are unconditionally
correlated \citep[see theorem 2 in][]{oneill2009}. In fact, latent
variables ensuring conditional independence are a common tool in
statistics to ``create'' or ``explain'' dependence between variables -
see, for instance, factor models \citep{joe2014, Bork2017, Nguyen2020}
and de Finetti's representation theorem \citep{oneill2009}. Therefore,
in the context of causal models, two variables causally unrelated may be
correlated if they are both affected by a third variable, called
confounder \citep{pearl2009, pearl2016}. For this reason, the analysis
of conditional and partial associations is of great importance in
statistical analysis of observational data; in the context of regression
analysis, this is done by adding control variables \citep{kutner2005},
and, in the context of correlation analysis, this is done through
partial correlation \citep{baba2004, kim2015}. A brief introduction on
how to choose control variables properly (from a causal viewpoint) is
presented in \citet{cinelli2022}.

In the same way that using a single numerical index (such as the
expectation) is not as informative as knowing the distribution of a
variable, it is also true that correlation coefficients fail to capture
several aspects of the dependence structure of a random vector - for
this reason, copula theory offers a much richer way to study statistical
dependence. Copulas are functions such that, given the copula and the
univariate (marginal) distributions of a random vector, the joint
distribution of the variables is fully characterized - therefore, they
fully characterize the dependence structure behind a random vector
\citep{nelsen2006, joe2014}. In fact, it has been proved that copulas
offer sufficient and necessary conditions for independence
\citep[Theorem 2.4.2 in][]{nelsen2006} and conditional independence
\citep[Theorem 1 in][]{gonzalezlopez2024}. \citet{embrechts2002}
presents an example of two bivariate distributions that share the same
marginal distributions and same Pearson's correlation value, but whose
tail dependence is not the same because the copulas are not the same.
This exemplifies the limitation of using Pearson's correlation
coefficient to summarize the dependence structure behind a random
vector. Therefore, having the notion of partial copula is of great
relevance to fully study the stochastic dependence between two variables
after controlling for some confounder, because they are a natural
extension of the notion of partial correlation.

Partial copulas were introduced in \citet{bergsma2004} as a tool to
construct a conditional independence test \citetext{\citealp[which has
also been explored
in][]{bergsma2011}; \citealp{patra2016}; \citealp[and][]{petersen2021}},
and their use in vine copula and stochastic process modeling has been
discussed in the literature \citep{bladt2022, bauer2016, spanhel2019}.
However, to the best of our knowledge, the use of this notion as a way
to formally study covariate-adjusted dependence has not been explored
enough in the literature \citep[despite the existence of some references
exploring related matters, see][ and \citet{spanhel2016}]{gijbels2015}.
In the context of linear models, the use of copulas to enable causal
inference was proposed in \citet{park2012} and further studied in
\citet{falkenstrom2023}; the method proposed by the authors is a
data-driven method that attempts to remove the confounder effect without
measuring the confounding variable directly (which may require strong
untestable assumptions). More recently, the use of vine copulas in
causal inference was considered by \citet{justus2024} and
\citet{czado2025} - however, to the best of our knowledge, the use of
partial copulas to investigate causal effects is not yet properly
explored in the literature.

Our main contributions are twofold. First, we give a practical
interpretation of the notion of partial copula as a nonlinear analogue
of partial correlation, suggesting that its use should not be limited to
conditional independence testing. Second, we prove several results
showing how properties of the conditional copulas - such as Kolmogorov
distance dependence, quadrant positive dependence and Spearman's
correlation - constrain the form of the partial copula. Additionally,
the results are illustrated through simulations.

The paper is organized as follows. Next section provides all the
necessary background to the study of partial copulas: Rosenblatt
transformations, probability integral transform and copula functions are
introduced. Section 3 is divided in three subsections, the first two
discussing the interpretation of partial copula and the last one
introducing new results connecting conditional and partial copulas. A
simulation study is conducted in Section 4, and final considerations are
given in Section 5.

\section{Background}\label{background}

\subsection{Copulas}\label{copulas}

Copulas are multivariate functions \([0,1]^d \to [0,1]\) used to
describe the dependence structure of random vectors of dimension
\(d \in \mathbb{N}\). The definition in the bivariate case is presented
below.

\begin{definition}\label{def:2copula}
A 2-copula is a function $C: [0,1]^2 \to [0,1]$ with the following properties:

(i) $C(u,0) = C(0,v) = 0$, and $C(u,1) = u$, $C(1,v) = v$, $\forall u,v: u,v \in [0,1]$;

(ii) for every $u_i, v_i \in [0,1]$, $i = 1,2$, such that $u_1 \leq u_2$ and $v_1 \leq v_2$,
$$C(u_2, v_2) - C(u_2, v_1) - C(u_1, v_2) + C(u_1, v_1) \geq 0.$$
\end{definition}

Note that Definition \ref{def:2copula} is equivalent to the definition
of 2-variate cumulative distribution function (CDF) whose marginal
distributions are uniform in \([0,1]\) range. Therefore, given any two
random variables \(U, V\) such that \(U \sim U(0,1)\) and
\(V \sim U(0,1)\), the cumulative distribution function
\(P(U \leq u, V \leq v)\) is necessarily a copula (because it satisfies
Definition \ref{def:2copula}).

Sklar's theorem \citep{sklar1959, nelsen2006} guarantees that, given any
two random variables \(X\) and \(Y\), it holds that
\[P(X \leq x, Y \leq y) = C(F_X(x), F_Y(y)),\] where \(F_X\) and \(F_Y\)
are, respectively, the cumulative distribution functions of \(X\) and
\(Y\), and \(C\) is a copula - if the variables are continuous, then
\(C\) is unique and it may be called ``the copula of \((X,Y)\)''. It's
possible to note that the copula of \((X,Y)\) is the CDF of
\((F_X(X), F_Y(Y))\).

The representation provided by Sklar's theorem shows that copulas are a
``marginal-free'' way to study the dependence between two (or more)
random variables, which makes clear that this is a central notion in the
study of stochastic dependence. For further arguments regarding the
importance of copula theory in the study of stochastic dependence, see
\citet{embrechts2002} and see also the monographs from
\citet{nelsen2006} and \citet{joe2014}.

\subsection{Probability integral transform and Rosenblatt
transforms}\label{probability-integral-transform-and-rosenblatt-transforms}

An important result in both copula theory and computational statistics
is the so-called ``probability integral transform'' \citep{int-transf},
presented below. This result ensures that a cumulative distribution
function makes a variable uniformly distributed when it is used as a
transformation applied in the random variable. Given that copulas are
cumulative distribution functions of variables uniformly distributed in
\([0,1]\), the importance of this result for copula theory seems
evident.

\begin{proposition}[Probability Integral Transform]
Let $F$ be an absolutely continuous cumulative distribution function, and let $F^{-1}$ be the inverse function of $F$. Then,
\begin{itemize}
\item If $U \sim U(0,1)$, then $F^{-1}(U)$ is a random variable with cumulative distribution function $F$ (which will be denoted as $F^{-1}(U) \sim F$).
\item Conversely, if $X \sim F$, then $F(X)$ is a random variable with uniform distribution on $(0,1)$, that is, $F(X) \sim U(0,1)$.
\end{itemize}
\end{proposition}

An even more general result is available in \citet{rosenblatt1952}:
given a sequence of random variables \(X_1, \ldots, X_n\), the random
variables
\[F_1(X_1), \, F_{X_2|X_1}(X_2|X_1), \, \ldots, \, F_{X_n|X_1,\ldots,X_{n-1}}(X_n|X_1, \ldots, X_{n-1})\]
are independent and uniformly distributed in the \([0,1]\) range. These
functions are usually called ``Rosenblatt transforms'' and they can be
seen as a kind of generalization of the probability integral transform
to the multivariate case. In statistical literature, those
transformations have already been employed to test conditional
independence \citep{bergsma2004, Song, patra2016} and, more recently, to
construct a nonparametric notion of error term in regression models
\citep{patra2016}. A central result connecting Rosenblatt transforms and
conditional independence is given by next proposition.

\begin{proposition}\citep{patra2016}
Let $(X, Y, Z)$ be an absolutely continuous random vector. Therefore, $X \perp\!\!\!\perp Y | Z$ if and only if $F_{X|Z}(X|Z)$, $F_{Y|Z}(Y|Z)$ and $Z$ are mutually independent.
\end{proposition}

This result suggests that studying the joint distribution of
\(F_{X|Z}(X|Z)\) and \(F_{Y|Z}(Y|Z)\) is a natural way to study the
dependence between \(X\) and \(Y\) adjusted by \(Z\), because those two
quantities would be independent in the particular case that they are
conditionally independent given \(Z\) (therefore, this joint
distribution would remove the ``spurious association'' induced by
\(Z\)). In the context of causal inference, if \(Z\) is the confounder
of \((X,Y)\), the two variables would be conditionally independent given
\(Z\), provided they are causally unrelated \citep[see considerations
about fork models in][]{pearl2009, pearl2016}.

\section{The partial copula}\label{the-partial-copula}

\subsection{Definition and
interpretation}\label{definition-and-interpretation}

This section presents the notion of partial copula, originally
introduced in \citet{bergsma2004}.

\begin{definition}
Let $(X, Y, Z)$ be an absolutely continuous random vector. The partial copula of $(X, Y)$ controlled by $Z$, which we denote by $C_{X,Y;Z}$, is defined as
$$C_{X,Y;Z}(x, y) = P(U_X\leq x, U_Y\leq y),$$
where $U_X= F_{X|Z}(X|Z)$ and $U_Y= F_{Y|Z}(Y|Z)$.
\end{definition}

\begin{remark}\label{rem:remark-2}
Note that $C_{X,Y;Z}$ fully satisfies the definition of copula (Definition \ref{def:2copula}), because $C_{X,Y;Z}$ is the cumulative distribution function of two random variables uniformly distributed in $[0,1]$.
\end{remark}

The partial copula \(C_{X,Y;Z}\) may be seen - in a similar way to
partial correlation - as a notion that represents the stochastic
dependence between \((X, Y)\) after removing the effect of a covariate
\(Z\). Further, differently from the notion of conditional copula,
\(C_{X,Y;Z}\) does not depend on the specification of a specific value
\(z\) of \(Z\).

Next result shows that the partial copula \(C_{X,Y;Z}\) can be written
as a mixture of all conditional copulas \(C_{X,Y|Z=z}\); this result, in
addition to providing an analytic representation for the partial copula,
highlights the capability of the concept in capturing conditional
dependence.

\begin{proposition}\label{prop:partial-copula}
Let $(X, Y, Z)$ be a continuous random vector. The partial copula $C_{X,Y;Z}$ can be written as
$$C_{X,Y;Z}(x, y) = \int_{-\infty}^{\infty} C_{X,Y|Z=z}(x, y) f_Z(z) dz,$$
where $C_{X,Y|Z=z}$ is the conditional copula of $(X, Y)$ given $Z = z$.
\end{proposition}

\begin{proof}
Note that
\begin{align*}
C_{X,Y;Z}(x, y) &= P(U_X\leq x, U_Y\leq y) \\
&= \int_{-\infty}^{\infty} P(U_X\leq x, U_Y\leq y | Z = z) f_Z(z) dz \\
&= \int_{-\infty}^{\infty} P(F_{X|Z}(X|Z) \leq x, F_{Y|Z}(Y|Z) \leq y | Z = z) f_Z(z) dz \\
&= \int_{-\infty}^{\infty} P(F_{X|Z}(X|z) \leq x, F_{Y|Z}(Y|z) \leq y | Z = z) f_Z(z) dz \\
&= \int_{-\infty}^{\infty} C_{X,Y|Z=z}(x, y) f_Z(z) dz,
\end{align*}
where the fourth equality comes from applying the substitution principle for conditional distributions \citep[Proposition 4.1 of][]{james1981}; and the last equality comes from applying Sklar's theorem \citep{sklar1959} in $P(F_{X|Z}(X|z) \leq x, F_{Y|Z}(Y|z) \leq y | Z = z)$.

This application of Sklar's theorem is justified because $P(F_{X|Z}(X|z) \leq x, F_{Y|Z}(Y|z) \leq y | Z = z)$ is a cumulative distribution function for any $z \in \mathbb{R}$; Furthermore, $F_{X|Z}(X|z)$ and $F_{Y|Z}(Y|z)$ are, respectively, strictly increasing transformations of $x$ and $y$, because they are cumulative distribution functions for any $z \in \mathbb{R}$; then, from Theorem 2.4.3 in \citep{nelsen2006}, the conditional copula identified by Sklar's theorem in $P(F_{X|Z}(X|z) \leq x, F_{Y|Z}(Y|z) \leq y | Z = z)$ is the same as the conditional copula of $P(X \leq x, Y \leq y | Z = z)$. Therefore, the last equality is proved taking into account that the univariate conditional distributions are uniform in $[0,1]$.
\end{proof}

Proof of Proposition 3 provides a formal justification for a similar
expression that appears in proof of Proposition 1 in
\citet{gijbels2015}.

\begin{remark}
A straightforward adaptation of the proof of Proposition 3 \citep[see also Equation 2.1 in][]{joe1996families} shows that, given a continuous random vector $(X,Y,Z)$,
\begin{align*}
P(X \leq x, Y \leq y) &= \int_{-\infty}^{\infty} P(X \leq x, Y \leq y | Z = z) f_Z(z) dz \\
&= \int_{-\infty}^{\infty} C_{X,Y|Z=z}(F_{X|Z}(x|z), F_{Y|Z}(y|z)) f_Z(z) dz.
\end{align*}
Therefore, the partial copula can also be seen as the distribution of $(X,Y)$ under the assumption that each variable is independent of $Z$ but keeping the same conditional copulas as in the true distribution of $(X,Y,Z)$.
\end{remark}

An interesting consequence of the last proposition - summarized in the
following corollary - is that, when the ``simplifying assumption''
\citep[commonly used in Vine copula modeling, see][]{czado2022} is
adopted, then the partial copula coincides with the conditional copula.

\begin{corollary}\label{cor:partial-conditional-copula}
Let $(X, Y, Z)$ be a continuous random vector. If the conditional copula of $(X, Y)$ given $Z = z$ is constant as a function of $z$ ("simplifying assumption"), then
$$C_{X,Y;Z}(x, y) = C_{X,Y|Z}(x, y).$$
\end{corollary}

\begin{proof}
This result is a direct consequence of integrating $C_{X,Y|Z=z}(x, y) f_Z(z)$ with respect to $z$ and taking into account that, from the theorem assumption, $C_{X,Y|Z=z}(x, y)$ is constant as a function of $z$.
\begin{align*}
C_{X,Y;Z}(x, y) &= \int_{-\infty}^{\infty} C_{X,Y|Z=z}(x, y) f_Z(z) dz \\
&= \int_{-\infty}^{\infty} C_{X,Y|Z}(x, y) f_Z(z) dz \\
&= C_{X,Y|Z}(x, y) \int_{-\infty}^{\infty} f_Z(z) dz \\
&= C_{X,Y|Z}(x, y).
\end{align*}
\end{proof}

From this result, it also becomes clear that when \((X, Y)\) are
conditionally independent given \(Z\) (which produces the conditional
copula \(C_{X,Y|Z=z}(x, y) = xy\)), then the partial copula is the
independence copula, which is coherent with lemma 3.1 of
\citet{patra2016}. However, the reciprocal statement is not true - in
the next remark, we construct a counterexample based on a model that
commonly appears in the literature as a counterexample for certain
statistical claims concerning pairwise independence or conditional
independence \citep[e.g.,][]{nelsen2012, gonzalezlopez2024}.

\begin{remark}\label{rem:partial-copula-necessary-not-sufficient}
Corollary \ref{cor:partial-conditional-copula} entails that, if $X \perp\!\!\!\perp Y | Z$, then $C_{X,Y;Z}(x, y) = xy$ (or, equivalently, $U_X$ and $U_Y$ are independent).

On the other hand, if $U_X$ and $U_Y$ are independent, it is not necessarily true that $X \perp\!\!\!\perp Y | Z$. Consider the following example: if $Z \sim U(0,1)$ and $C_{X,Y|Z=z}(x, y) = xy + (2z - 1)xy(1-x)(1-y)$ (that is, a Farlie-Gumbel-Morgenstern copula with parameter $2z - 1$), then from Proposition \ref{prop:partial-copula},
$$C_{X,Y;Z}(x, y) = \int_0^1 [xy + (2z-1)xy(1-x)(1-y)] dz = xy,$$
which is the independence copula, which entails, from the definition of the partial copula, that $U_X$ and $U_Y$ are independent. However, except for $Z = 0.5$, $X$ and $Y$ are not conditionally independent given $Z = z$. Note that if $X \perp\!\!\!\perp Z$, $Y \perp\!\!\!\perp Z$ and $X,Y \sim U(0,1)$, then this example produces a trivariate Farlie-Gumbel-Morgenstern copula .
\end{remark}

Therefore, \(U_X\perp\!\!\!\perp U_Y\) is a necessary but not sufficient
condition for \(X \perp\!\!\!\perp Y | Z\). A possible explanation for
this fact is that, in the example presented in the remark, the
conditional association of \((X, Y)\) given \(Z = z\) is positive for
some values of \(z\) and negative for other values of \(z\). In this
sense, the partial copula represents a kind of ``average conditional
association''; this fact is shown in next remark.

\begin{remark}\label{rem:partial-copula-expected-conditional-copula}
Let $g(x, y, z) = C_{X,Y|Z=z}(x, y)$. $g(x, y, Z)$ is a random variable, and Proposition \ref{prop:partial-copula} entails that
$$C_{X,Y;Z}(x, y) = E[g(x, y, Z)],$$
which means that the partial copula is the expected conditional copula.
\end{remark}

\subsection{Partial copula as an extension of partial
correlation}\label{partial-copula-as-an-extension-of-partial-correlation}

Next remark underscores the connection between the partial copula and
the notion of partial correlation.

\begin{remark}\label{rem:partial-correlation}
A possible way to define the partial correlation between two random variables $X$ and $Y$, given $Z$, is to construct two regression models $X = \alpha + \beta Z + \epsilon_{X;Z}$ and $Y = \gamma + \theta Z + \epsilon_{Y;Z}$ and then compute the Pearson's correlation between the error terms $\epsilon_{X;Z}$ and $\epsilon_{Y;Z}$ \citep{kim2015}.

More recently, \cite{patra2016} studied the problem of how to construct a nonparametric notion of error term, that is, a notion of error that does not assume a specific functional form for the relationship between the regressor and the response variable. Imposing some conditions that are reasonable to expect from a notion of error term (such as being independent from the regressor), \cite{patra2016} showed that the Rosenblatt transforms $F_{X|Z}(X|Z)$ and $F_{Y|Z}(Y|Z)$ satisfy those conditions and, therefore, may be seen as a general notion of error term \citep[for further considerations, see also Section 4 of][]{justus2024}.

Partial copulas are the joint CDF of the nonparametric error terms $F_{X|Z}(X|Z)$ and $F_{Y|Z}(Y|Z)$, and, therefore, can be seen as a nonparametric extension of the notion of partial correlation.
\end{remark}

Previous results in the literature have established that zero partial
correlation is not equivalent to conditional independence, except in
some special cases such as the multivariate normal distribution
\citep{baba2004}. Next example shows that, even under conditional
independence, the partial correlation may be arbitrarily close to 1,
which suggests that the partial correlation may be a very poor measure
of covariate-adjusted association.

\begin{example}\label{ex:partial-correlation}
Let $Z \sim N(0,1)$. Let $X|Z=z \sim N(z^2, \sigma^2)$ and $Y|Z=z \sim N(z^2, \sigma^2)$ be conditionally independent given $Z=z$. It is straightforward to see that, in this case,
$$Corr(X,Z) = Corr(Y,Z) = 0,$$
which entails, from the definition of partial correlation, that $Corr(X,Y;Z) = Corr(X,Y)$. On the other hand, $X$ and $Y$ are exchangeable and, from Theorem 2 in \cite{oneill2009}, it holds that
$$Corr(X,Y) =  \frac{Var(E(X|Z))}{Var(E(X|Z)) + E(Var(X|Z))} = \frac{Var(Z^2)}{Var(Z^2) + \sigma^2}.$$
Therefore,
$$Corr(X,Y;Z) = Corr(X,Y) = \frac{2}{2 + \sigma^2} > 0,$$
which becomes arbitrarily close to 1 as $\sigma^2$ approaches 0 - that is, $\mathsf{lim}_{\sigma^2 \to 0} Corr(X,Y;Z) = 1$ -, even though $X$ and $Y$ are conditionally independent given $Z$.
\end{example}

The result seen in Example \ref{ex:partial-correlation} may look
worrying: even when two variables are conditionally independent given a
covariate, their partial correlation may be close to 1. An intuitive
explanation for this fact may come from Remark
\ref{rem:partial-correlation}: given that the partial correlation is the
Pearson's correlation between the error terms of two linear regression
models, only the ``linear effect'' of \(Z\) on the variables is being
``removed''. In Example \ref{ex:partial-correlation}, the expectation of
each variable given \(Z\) is not linear (not even monotonic), and,
therefore, the partial correlation fails in removing the effect of \(Z\)
on the variables. On the other hand, if \(X\) and \(Y\) are
conditionally independent given \(Z\), then the partial copula is always
the independence copula. In this sense, the partial copula fully removes
the ``spurious association'' induced by \(Z\). Therefore, except when a
researcher has strong reasons to make some parametric assumptions (such
as linearity), the partial copula looks to be a much more appropriate
notion for representing covariate-adjusted statistical association.

\subsection{Effect of properties of the conditional copulas on the
partial
copula}\label{effect-of-properties-of-the-conditional-copulas-on-the-partial-copula}

Given the functional relationship between the partial and the
conditional copulas, an interesting research line that emerges is
studying to what extent properties of the conditional copulas constrain
the properties of the partial copula. A first result on this matter is
given by next proposition. The result is based on the notions of
quadrant positive and quadrant negative dependence: a copula \(C\) is
said to have quadrant positive dependence (QPD) if \(C(u, v) \geq uv\)
for all \((u, v) \in [0,1]^2\). In an analogous way, the copula has
quadrant negative dependence (QND) if \(C(u, v) \leq uv\) for all
\((u, v) \in [0,1]^2\) \citep{joe2014}. A direct consequence of the fact
that the partial copula is an expected conditional copula is given by
the following proposition.

\begin{proposition}\label{prop:quadrant-dependence}
Let $(X, Y, Z)$ be a continuous random vector. If the conditional copula of $(X, Y)$ given $Z = z$ has positive quadrant dependence for all $z \in \mathbb{R}$, then the partial copula $C_{X,Y;Z}(x, y)$ has positive quadrant dependence.

In the same way, if the conditional copula of $(X, Y)$ given $Z = z$ has negative quadrant dependence for all $z \in \mathbb{R}$, then the partial copula $C_{X,Y;Z}(x, y)$ has negative quadrant dependence.
\end{proposition}

\begin{proof}
If $g(x, y, z) = C_{X,Y|Z=z}(x, y) \geq xy$ for all $z \in \mathbb{R}$, then, by interpreting the constant $xy$ as a degenerate random variable and applying the monotonicity of expectation \citep[property E2 of section 3.3 of][]{james1981}, we have the inequality
$$xy \leq E[g(x, y, Z)] = C_{X,Y;Z}(x, y).$$
Therefore, $C_{X,Y;Z}(x, y) \geq xy$, as we wished to prove. The same argument holds for negative quadrant dependence.
\end{proof}

An intuitive interpretation for Proposition
\ref{prop:quadrant-dependence} is that, whenever the conditional
association between \(X\) and \(Y\) is positive {[}negative{]} for all
values of \(Z\), the partial copula also exhibits positive
{[}negative{]} dependence.

Next proposition uses a dependence measure based on Kolmogorov distance,
which, for simplicity, we will refer to it as Kolmogorov distance
dependence (KDD). This coefficient was introduced in \citet{wolff1977}
and \citet{schweizer1981}, and is defined as
\[D(X, Y) = 4 \sup_{x,y \in [0,1]} (|C_{X,Y}(x, y) - xy|).\] Note that
the behavior of KDD differs from the behavior of the so-called
``concordance measures'', such as Spearman's correlation coefficient,
which may be equal to zero even if the two variables are not independent
\citep{garciaportugues2025}. This happens intuitively because dependence
and concordance measures have two different goals: the former ones
measure deviation from independence, and the last ones measure the
degree of monotonic association \citep[see][]{nelsen2006, joe2014}.

It is immediate to see that KDD is proportional to the Kolmogorov
distance between a given copula and the independence copula; the
constant 4 ensures that \(D(X, Y) \in [0, 1]\). This coefficient is
constructed intending to measure in what extent a copula deviates from
independence, without distinguishing between ``positive'' and
``negative'' dependence; in this sense, values closer to 0 represent
independence and cases closer to 1 represent strong dependence - for
further details, see \citet{wolff1977} and \citet{schweizer1981}. The
following proposition shows how the value of KDD in the conditional
copulas constrains the value of KDD in the partial copula.

\begin{proposition}\label{prop:partial-copula-independence}
Let $X, Y, Z$ be jointly continuous random variables. Then,
$$D(U_X, U_Y) \leq \sup_{z \in \mathbb{R}} (D(X, Y | Z = z)).$$
\end{proposition}

\begin{proof}
Suppose $\sup_{z \in \mathbb{R}}(D(X, Y | Z = z)) = k$, for some $k \in [0,1]$. This means that, for every $z \in \mathbb{R}$, it holds that
$$k \geq D(X, Y | Z = z) = 4 \sup_{x,y \in [0,1]} (|C_{X,Y|Z=z}(x, y) - xy|),$$
which entails that $|C_{X,Y|Z=z}(x, y) - xy| \leq k/4$ for every $x, y \in [0, 1]$, $z \in \mathbb{R}$. Therefore,
$$C_{X,Y|Z=z}(x, y) \in \left[xy - \frac{k}{4}, xy + \frac{k}{4}\right].$$

The remainder of the proof comes from remembering that the partial copula evaluated in a point $(x, y) \in [0, 1]^2$ is the expectation of a random variable (see Remark \ref{rem:partial-copula-expected-conditional-copula}), which enables using the monotonicity of expectation: property E2 of section 3.3 of \citet{james1981} readily entails that, if a random variable belongs almost surely to an interval, then the same happens to its expectation. Therefore,
$$xy - \frac{k}{4} \leq C_{X,Y;Z}(x, y) \leq xy + \frac{k}{4},$$
for every $x, y \in [0, 1]$. Therefore, $|C_{X,Y;Z}(x, y) - xy| \leq k/4$, which entails that
$$D(U_X, U_Y) = 4 \sup_{x,y \in [0,1]} (|C_{X,Y;Z}(x, y) - xy|) \leq k.$$
\end{proof}

The last result can be intuitively interpreted in the following way:
regardless of simplifying assumption holding or not, if all the
conditional copulas are distant from independence at most \(k\), then
the partial copula is distant from independence at most \(k\) too. This
entails that, if all conditional copulas are close to the independence
case, then the partial copula is necessarily close to independence.

Furthermore, it is possible to also bound Spearman's correlation
coefficient (here denoted by the letter \(\rho\)) in the partial copula,
provided an upper bound for the KDD of the conditional copulas is also
known (it is worth remembering that Spearman's correlation is also a
functional of the copula, see \citet{fredricks2007}).

\begin{proposition}\label{def:spearman-bounds}
Let $X, Y, Z$ be jointly continuous random variables. If $(D(X, Y | Z = z)) \leq k$ for all $z \in \mathbb{R}$, then
$$\max(-1, -3k) \leq \rho(X, Y; Z) \leq \min(1, 3k).$$
\end{proposition}

\begin{proof}
The proof of Proposition \ref{prop:partial-copula-independence} entails that, if $(D(X, Y | Z = z)) \leq k$ for all $z \in \mathbb{R}$, then
$$xy - \frac{k}{4} \leq C_{X,Y;Z}(x, y) \leq xy + \frac{k}{4},$$
which entails that
$$\int_0^1 \int_0^1 C_{X,Y;Z}(x, y) dx dy \in \left[\frac{1-k}{4}, \frac{1+k}{4}\right].$$

Theorem 5.1.6 in \citet{nelsen2006} states that, given a copula $C$, its Spearman's correlation can be represented as $12 \int_0^1 \int_0^1 C_{X,Y;Z}(x, y) dx dy - 3$. Applying this representation in the partial copula, it holds that
$$-3k \leq \rho(X, Y; Z) \leq 3k,$$
as we wished to prove. The final bound in the proposition is obtained taking into account that Spearman's coefficient belongs to the range $[-1, 1]$ \citep[see Theorem 5.1.9 in][]{nelsen2006}.
\end{proof}

An even better bound can be obtained for the Kendall's tau correlation
coefficient (denoted by \(\tau\)), which is also a functional of the
copula \citep{nelsen2006} - see next Proposition.

\begin{proposition}
Let $X, Y, Z$ be jointly continuous random variables. If $(D(X, Y | Z = z)) \leq k$ for all $z \in \mathbb{R}$, then
$$\max(-1, -2k) \leq \tau(X, Y; Z) \leq \min(1, 2k).$$
\end{proposition}

\begin{proof}
We have seen in the proof of Proposition \ref{prop:partial-copula-independence} that, if $(D(X, Y | Z = z)) \leq k$ for all $z \in \mathbb{R}$, then
$$xy - \frac{k}{4} \leq C_{X,Y;Z}(x, y) \leq xy + \frac{k}{4}.$$
Therefore, inserting $U_X$ and $U_Y$ as arguments, we have
$$U_X U_Y - \frac{k}{4} \leq C_{X,Y;Z}(U_X, U_Y) \leq U_X U_Y + \frac{k}{4}.$$
The monotonicity of expectation \citep[property E2 of section 3.3 of][]{james1981} entails that
$$E[U_X U_Y] - \frac{k}{4} \leq E[C_{X,Y;Z}(U_X, U_Y)] \leq E[U_X U_Y] + \frac{k}{4}.$$
The Equation 5.1.15b from Theorem 5.1.6 in \citet{nelsen2006} entails that $\rho(X,Y;Z) = 12 E[U_X U_Y] -3$, which means that
$$\frac{\rho(X,Y;Z) + 3}{12} - \frac{k}{4} \leq E[C_{X,Y;Z}(U_X, U_Y)] \leq \frac{\rho(X,Y;Z) + 3}{12} + \frac{k}{4}.$$
The final argument comes from remembering that, according to Theorem 5.1.3 in \citet{nelsen2006}, Kendall's tau correlation coefficient can be represented as $4E[C_{X,Y;Z}(U_X, U_Y)] - 1$, which leads us to the conclusion that
$$\frac{\rho(X,Y;Z) + 3}{3} - k -1 \leq \tau(X, Y; Z) \leq \frac{\rho(X,Y;Z) + 3}{3} + k -1,$$
and, from the inequality $-3k \leq \rho(X, Y; Z) \leq 3k,$ it follows that
$$-2k \leq \tau(X, Y; Z) \leq 2k.$$
As in the previous proof, the final bounds are obtained using the fact that Kendall's coefficient belongs to the range $[-1, 1]$ \citep[see Theorem 5.1.9 in][]{nelsen2006}.
\end{proof}

Additionally, the connection between the correlation in the conditional
copulas and in the partial copula can be summarized in the following
theorem.

\begin{proposition}\label{prop:spearman-partial-conditional}
Let $X, Y, Z$ be jointly continuous random variables. Then,
$$\rho(X, Y; Z) = \int_{-\infty}^{\infty} \rho(X, Y | Z = z) f_Z(z) dz.$$
\end{proposition}

\begin{proof}
Applying Equation 5.1.15c of \citet{nelsen2006} in the partial copula, we obtain
$$\rho(X, Y; Z) = 12 \int_0^1 \int_0^1 C_{X,Y;Z}(x, y) dx dy - 3.$$
Now using representation from Proposition \ref{prop:partial-copula}, we obtain
\begin{align*}
\rho(X, Y; Z) &= 12 \int_0^1 \int_0^1 \int_{-\infty}^{\infty} C_{X,Y|Z=z}(x, y) f_Z(z) dz dx dy - 3 \\
&= 12 \int_{-\infty}^{\infty} \left(\int_0^1 \int_0^1 C_{X,Y|Z=z}(x, y) dx dy \right) f_Z(z) dz - 3 \\
&= 12 \int_{-\infty}^{\infty} \left(\frac{\rho(X, Y | Z = z) + 3}{12}\right) f_Z(z) dz - 3 \\
&= \int_{-\infty}^{\infty} \rho(X, Y | Z = z) f_Z(z) dz.
\end{align*}
\end{proof}

That is, the Spearman's correlation induced by the partial copula is the
expected conditional Spearman's correlation - this greatly clarifies the
role of conditional correlations in determining partial correlation.
Finally, a consequence of this result in a family of copulas studied in
\citet{rodriguezlallena2004} is illustrated in the following example.

\begin{example}
\cite{rodriguezlallena2004} studied a nonparametric family of copulas where the copula of two random variables $X',Y'$ can be written as $C_{X',Y'}(x, y) = xy + f(x)g(y)$. Now, consider a case where all the conditional copulas of $X,Y$ given $Z$ can be written as
$$C_{X,Y|Z=z}(x, y) = xy + f(x,z)g(y,z),$$
which means that all the conditional copulas belong to the family of copulas studied in \cite{rodriguezlallena2004}. \cite{komelj2010} showed that all copulas in this family have Spearman's correlation coefficient belonging to $[-3/4, 3/4]$. Therefore, applying Proposition \ref{prop:spearman-partial-conditional} and monotonicity of expectation \citep[property E2 of section 3.3 of][]{james1981}, it is straightforward to see that
$$-\frac{3}{4} \leq \rho(X, Y; Z) \leq \frac{3}{4}.$$
Therefore, the fact that all the conditional copulas belong to the family of copulas studied in \cite{rodriguezlallena2004} entails that the Spearman's correlation coefficient in the partial copula is also bounded by $[-3/4, 3/4]$.
\end{example}

\section{Simulation Study}\label{simulation-study}

The theoretical results in Section 3 establish that the partial copula
\(C_{X,Y;Z}\) captures the dependence between \(X\) and \(Y\) after
removing the effect of \(Z\). In particular, Proposition
\ref{prop:spearman-partial-conditional} shows that the Spearman
correlation induced by the partial copula equals the expected
conditional correlation, while Proposition
\ref{prop:quadrant-dependence} establishes that the sign of quadrant
dependence is preserved when it is constant across all conditional
copulas. In this section, we conduct a simulation study to illustrate
these results and to demonstrate how the partial copula can reveal the
true conditional association structure when marginal correlations are
confounded. All simulations were performed in R using the
\texttt{rvinecopulib} package \citep{rvinecopulib}. The manuscript and
all analyses are fully reproducible through the \texttt{rix} package
\citep{rix}, and the source code and environment files are available in
the accompanying GitHub repository:
\url{https://github.com/felipelfv/copula_project}.

\subsection{Design}\label{design}

We employ a C-vine copula with \(Z\) as the root node, which decomposes
the joint distribution of \((X, Y, Z)\) through the pair copulas
\(C_{X,Z}\), \(C_{Y,Z}\), and \(C_{X,Y|Z}\). For each configuration, we
generate a single sample of \(n = 5000\) observations
\((X_i, Y_i, Z_i)_{i=1}^{n}\) and compute the partial copula
pseudo-observations via h-functions,
\[U_{X,i} = h_2(X_i, Z_i; C_{X,Z}) \quad \text{and} \quad U_{Y,i} = h_2(Y_i, Z_i; C_{Y,Z}),\]
where \(h_2(u, v; C) = \partial C(u,v)/\partial v\) is the conditional
distribution function of the first argument given the second.

To assess how confounding and conditional dependence interact, we
compare the marginal and partial rank correlations. Specifically, we
compute Spearman's \(\rho\) and Kendall's \(\tau\) for both the original
pair \((X, Y)\) and the partial copula pair \((U_X, U_Y)\). A
discrepancy between these quantities indicates confounding: dependence
between \(X\) and \(Y\) that is induced by \(Z\) rather than by direct
conditional association.

We consider eleven scenarios, summarized in Table 1. The triplet
\((s_{XZ}, s_{YZ}, s_{XY|Z})\) indicates the sign of dependence in each
pair copula: \(+\) for positive, \(-\) for negative, and \(\text{ind}\)
for independence. Scenarios 1--10 satisfy the simplifying assumption, so
that \(C_{X,Y|Z=z}\) does not vary with \(z\) and Corollary
\ref{cor:partial-conditional-copula} applies. Scenario 11 violates the
simplifying assumption by letting the conditional copula parameter vary
as a function of \(Z\) through specifications \(\theta(z) = \exp(z)\),
\(\theta(z) = -\exp(z)\), and \(\theta(z) = 1 - 2z\); the last of these
illustrates Remark \ref{rem:partial-copula-necessary-not-sufficient},
where positive and negative conditional associations cancel out in
expectation.

Across scenarios, we consider multiple copula families (Frank, Gumbel,
Clayton, and Gaussian) over a range of parameter values to ensure
robustness of the findings. Negative dependence for Gumbel and Clayton
copulas is achieved via 90-degree rotation.

\small

\begin{longtable}[]{@{}ccl@{}}
\caption{Summary of simulation scenarios. The triplet
\((s_{XZ}, s_{YZ}, s_{XY|Z})\) denotes the sign of dependence in each
pair copula: \(+\) positive, \(-\) negative, \(\text{ind}\)
independence, and \(\cdot\) varies across
configurations.}\tabularnewline
\toprule\noalign{}
Scenario & Structure & Description \\
\midrule\noalign{}
\endfirsthead
\toprule\noalign{}
Scenario & Structure & Description \\
\midrule\noalign{}
\endhead
\bottomrule\noalign{}
\endlastfoot
1 & \((+,+,+)\) & Positive confounding, positive conditional \\
2 & \((+,+,\text{ind})\) & Positive confounding, conditional indep. \\
3 & \((-,-,\text{ind})\) & Negative confounding, conditional indep. \\
4 & \((+,-,\text{ind})\) & Opposite confounding, conditional indep. \\
5 & \((+,+,-)\) & Positive confounding, negative conditional \\
6 & \((-,-,+)\) & Negative confounding, positive conditional \\
7 & \((-,-,-)\) & Negative confounding, negative conditional \\
8 & \((+,-,+)\) & Opposite confounding, positive conditional \\
9 & \((+,-,-)\) & Opposite confounding, negative conditional \\
10 & \((\text{ind},\text{ind},\cdot)\) & No confounding \\
11 & --- & Simplifying assumption violated \\
\end{longtable}

\normalsize

\subsection{Results}\label{results}

Table 2 in the Appendix reports Spearman's \(\rho\) and Kendall's
\(\tau\) for both the marginal pair \((X,Y)\) and the partial copula
pair \((U_X, U_Y)\) across all scenarios. The results confirm the
theoretical predictions. In conditional independence scenarios (2, 3,
and 4), the partial correlations are negligible while marginal
correlations can be substantial, demonstrating that the partial copula
successfully removes confounding. In scenarios where confounding
reinforces conditional dependence (1, 6, and 9), the partial
correlations correctly recover the conditional association, which is
attenuated relative to the marginal correlations. Notably, scenarios 5
and 8 exhibit Simpson's paradox: marginal and partial correlations have
opposite signs. In these scenarios, the confounding effect opposes the
conditional dependence --- for instance, in scenario 5 \((+,+,-)\), the
product of two positive marginal dependencies induces a positive
spurious correlation that counteracts the negative conditional
dependence. This illustrates how the partial copula can reveal
conditional associations that are masked or reversed by confounding.
Finally, Scenario 11 examines violations of the simplifying assumption.
In cases 11a and 11b, where \(\theta(z) = \exp(z)\) and
\(\theta(z) = -\exp(z)\) respectively, the sign of conditional
dependence remains constant across all values of \(Z\), and the partial
correlations closely match the marginal correlations, consistent with
Proposition \ref{prop:quadrant-dependence}. In contrast, case 11c
deserves particular attention. Here the conditional copula parameter
\(\theta(z) = 1 - 2z\) is a Gaussian copula whose correlation switches
sign at \(z = 0.5\): for \(z < 0.5\) the conditional dependence is
positive, while for \(z > 0.5\) it becomes negative. As a result, the
partial copula --- which averages over all conditional copulas ---
yields near-zero partial correlation despite strong conditional
dependence for almost all values of \(Z\), confirming Remark
\ref{rem:partial-copula-necessary-not-sufficient}. Figure
\ref{fig:scenario-11c-detail} illustrates this by displaying the
conditional scatter plots for \(Z < 0.5\) and \(Z > 0.5\) alongside the
partial copula: the opposing dependence structures are clearly visible
in the conditional panels but cancel out in the partial copula. This
highlights a fundamental limitation: the partial copula captures average
conditional dependence and may not be the best tool when one is
interested in local, \(z\)-specific dependence.

\begin{figure}

{\centering \includegraphics[width=1\linewidth]{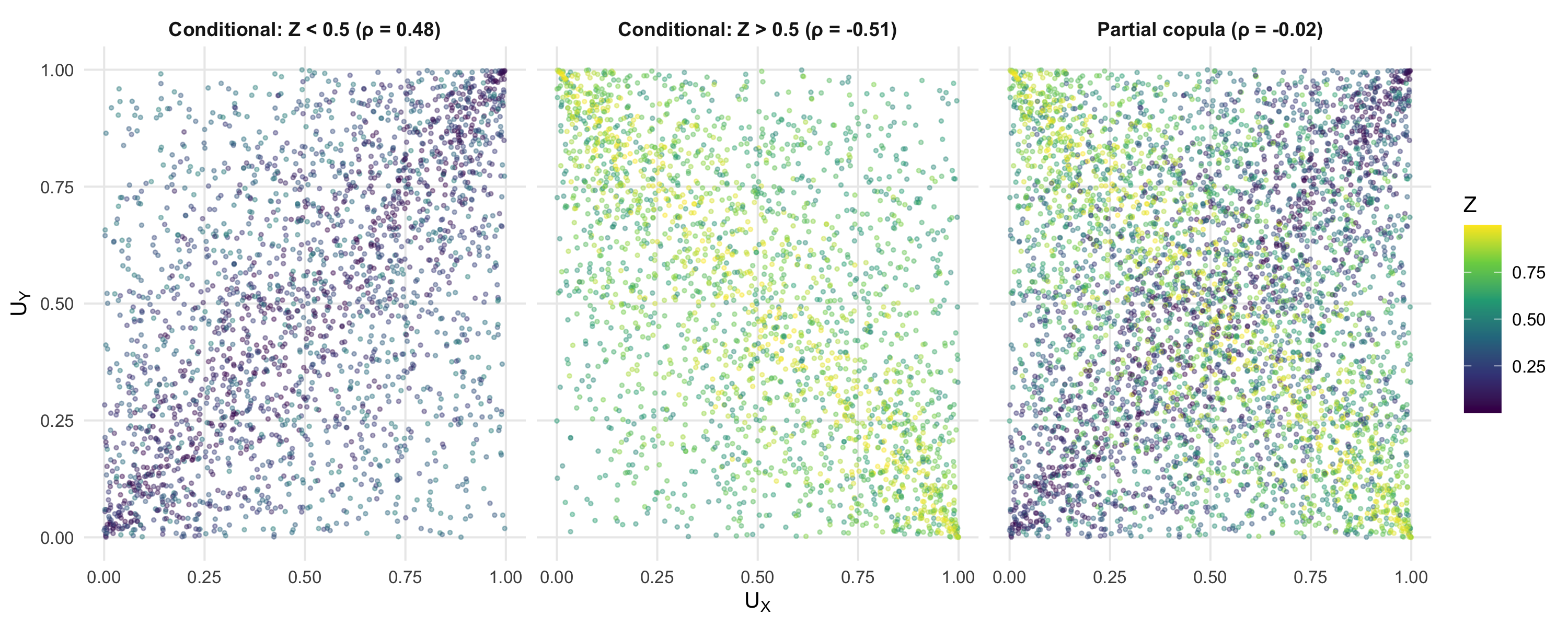} 

}

\caption{Scenario 11c: Gaussian copula with $\theta(z) = 1 - 2z$. Left: conditional scatter for $Z < 0.5$ (positive dependence). Center: conditional scatter for $Z > 0.5$ (negative dependence). Right: partial copula averaging over all $Z$. Color gradient represents $Z$; $\rho$ denotes Spearman's correlation.}\label{fig:scenario-11c-detail}
\end{figure}

\section{Discussion}\label{discussion}

In this paper, we discussed the notion of partial copula, originally
introduced in \citet{bergsma2004} to test conditional independence. This
test was introduced upon the observation that, under conditional
independence, the variables \(F_{X|Z}(X|Z)\) and \(F_{Y|Z}(Y|Z)\),
usually known as Rosenblatt transformations, are independent. Although
we recognize the usefulness of the partial copula as a tool to test
conditional independence, our main focus has been to highlight its
broader applicability as a general representation of conditional
association, which has not been well explored previously in the
literature, to the best of our knowledge. The starting point of this
paper is the observation that the variables \(F_{X|Z}(X|Z)\) and
\(F_{Y|Z}(Y|Z)\) are uniformly distributed in \([0,1]\), which means
that their joint cumulative distribution function is a copula, which
enables studying several dependence concepts constructed in the copula
literature, such as Spearman's and Kendall's coefficients.

First, we showed in Proposition 3 and Remark 2 that the partial copula
of \((X,Y)\) adjusted by \(Z\) may be viewed as the joint distribution
of \((X,Y)\) under the assumption that \(X\) and \(Y\) are both
independent of \(Z\) (marginally) but preserving their conditional
copula. This interpretation clarifies the role of the partial copula as
a general representation of conditional association, since it captures
the dependence between \(X\) and \(Y\) after removing the effect of
\(Z\). Second, we established several properties of the partial copula,
including bounds for Spearman's and Kendall's coefficients in the
partial copula when the maximum value of the Kolmogorov distance
dependence in the conditional copulas is known (Propositions 6 and 7),
the connection between conditional and partial correlations (Proposition
8), and a characterization of the partial copula when the conditional
copulas belong to a specific parametric family (Example 2). At an
intuitive level, those results show that the partial copula captures the
properties of conditional copulas while removing the effect of \(Z\).

In the context of causal inference, if \(Z\) is a confounding variable,
the dependence between \(X\) and \(Y\) adjusted by \(Z\) - as
traditionally estimated in linear models and other statistical
techniques - may be interpreted as the true causal connection between
the two variables. The properties proved in this paper underscore the
capability of the partial copula to capture this causal connection
without requiring restrictive assumptions such as linearity or the
simplifying assumption used in vine copula modeling.

Finally, a simulation study was conducted to compare how the partial and
the marginal copulas may be different. In scenarios presenting
conditional independence (scenarios 2, 3 and 4), the partial copula is
the independence copula, while the marginal copula contains the spurious
association induced by \(Z\). In scenarios where the confounding effect
reinforces the conditional dependence (1, 6, and 9), the association
observed in the marginal copula was stronger than in the partial copula.
In scenarios where the confounding effect opposes the conditional
dependence (5, 7, and 8), Simpson's paradox may arise, which means
marginal and partial correlations having opposite signs. In scenarios
with no confounding (10), the marginal and the partial copula are the
same. Finally, scenario 11 illustrated cases violating the simplifying
assumption, which shows a fundamental difference between partial and
conditional copulas: conditional copulas depend on the specification of
the value of \(Z\) being conditioned on, while the partial copula does
not. In this sense, the partial copula may be seen as a ``summary'' of
the conditional copulas across all values of \(Z\).

The diversity of scenarios simulated shows that, in several contexts,
the marginal association and the average conditional association may be
very different - including scenarios known as Simpson's paradox, where
even the signs of correlation coefficients are different. In causal
inference literature \citep{pearl2009, pearl2016}, it is well known that
marginal associations may be very different from causal effects due to
confounding. The results in this paper show that the partial copula is a
useful tool to estimate the average conditional association, which may
be interpreted as a causal effect under standard assumptions in causal
inference. From this paper, we hope to have shown both the importance of
using conditional association instead of marginal association in several
contexts, and the capability of the partial copula to capture this
conditional association in a general way.

Some directions for future research include studying directed dependence
coefficients \citep{azadkia2021simple, fuchs} in the context of partial
copulas, and exploring the use of partial copulas in the context of
causal discovery algorithms. Furthermore, further advances on the
estimation of partial copulas from real data are of great interest for
the scientific community, since it would allow researchers to apply the
insights of this paper in empirical contexts. \citet{bergsma2004}
discussed the use of kernel estimators to estimate the conditional CDFs;
however further exploration on the properties of such estimators seems
necessary. Although there is a rich literature in the estimation of
conditional copulas \citep[e.g.,][]{VERAVERBEKE2011, math14050914}, the
estimation of partial copulas itself is a subject that requires further
developments.

\section{Acknowledgment}\label{acknowledgment}

The authors thank Professor Concepción Ausín Olivera (Carlos III
University of Madrid, Spain) for the feedback given on a previous
version of this manuscript, which helped improve this work.

\newpage

\bibliography{references}

\newpage

\section*{Appendix}\label{appendix}
\addcontentsline{toc}{section}{Appendix}

\begingroup\fontsize{8}{10}\selectfont

\begin{longtable}[t]{llcccc}
\caption{\label{tab:results-table}Spearman's $\rho$ and Kendall's $\tau$ for the marginal pair $(X, Y)$ and the partial copula pair $(U_X, U_Y)$ across all simulated scenarios. The Specification column reports the copula family and the parameter triplet $\theta = (\theta_{XZ}, \theta_{YZ}, \theta_{XY|Z})$; (R) indicates a 90-degree rotation and ind denotes independence. P-values are reported in parentheses next to the corresponding correlation coefficient.}\\
\toprule
\multicolumn{2}{c}{ } & \multicolumn{2}{c}{Marginal $(X, Y)$} & \multicolumn{2}{c}{Partial $(U_X, U_Y)$} \\
\cmidrule(l{3pt}r{3pt}){3-4} \cmidrule(l{3pt}r{3pt}){5-6}
Case & Specification & ρ & τ & ρ & τ\\
\midrule
\endfirsthead
\caption[]{Spearman's $\rho$ and Kendall's $\tau$ for the marginal pair $(X, Y)$ and the partial copula pair $(U_X, U_Y)$ across all simulated scenarios. The Specification column reports the copula family and the parameter triplet $\theta = (\theta_{XZ}, \theta_{YZ}, \theta_{XY|Z})$; (R) indicates a 90-degree rotation and ind denotes independence. P-values are reported in parentheses next to the corresponding correlation coefficient. \textit{(continued)}}\\
\toprule
\multicolumn{2}{c}{ } & \multicolumn{2}{c}{Marginal $(X, Y)$} & \multicolumn{2}{c}{Partial $(U_X, U_Y)$} \\
\cmidrule(l{3pt}r{3pt}){3-4} \cmidrule(l{3pt}r{3pt}){5-6}
Case & Specification & ρ & τ & ρ & τ\\
\midrule
\endhead

\endfoot
\bottomrule
\endlastfoot
\addlinespace[0.3em]
\multicolumn{6}{l}{\textbf{Scenario 1: (+, +, +)}}\\
\hspace{1em}1a & Frank, θ = (1, 1, 1) & 0.174 (p < 0.001) & 0.117 (p < 0.001) & 0.152 (p < 0.001) & 0.102 (p < 0.001)\\
\hspace{1em}1b & Frank, θ = (3, 3, 1) & 0.326 (p < 0.001) & 0.220 (p < 0.001) & 0.164 (p < 0.001) & 0.110 (p < 0.001)\\
\hspace{1em}1c & Frank, θ = (5, 5, 1) & 0.495 (p < 0.001) & 0.341 (p < 0.001) & 0.157 (p < 0.001) & 0.105 (p < 0.001)\\
\hspace{1em}1d & Gumbel, θ = (1.5, 1.5, 1.5) & 0.614 (p < 0.001) & 0.442 (p < 0.001) & 0.479 (p < 0.001) & 0.336 (p < 0.001)\\
\hspace{1em}1e & Gumbel, θ = (2, 2, 1.5) & 0.731 (p < 0.001) & 0.546 (p < 0.001) & 0.483 (p < 0.001) & 0.338 (p < 0.001)\\
\hspace{1em}1f & Gumbel, θ = (3, 3, 1.5) & 0.860 (p < 0.001) & 0.681 (p < 0.001) & 0.473 (p < 0.001) & 0.331 (p < 0.001)\\
\hspace{1em}1g & Clayton, θ = (1, 1, 1) & 0.618 (p < 0.001) & 0.446 (p < 0.001) & 0.494 (p < 0.001) & 0.346 (p < 0.001)\\
\hspace{1em}1h & Clayton, θ = (3, 3, 1) & 0.814 (p < 0.001) & 0.630 (p < 0.001) & 0.476 (p < 0.001) & 0.332 (p < 0.001)\\
\hspace{1em}1i & Clayton, θ = (5, 5, 1) & 0.891 (p < 0.001) & 0.725 (p < 0.001) & 0.468 (p < 0.001) & 0.325 (p < 0.001)\\
\addlinespace[0.3em]
\multicolumn{6}{l}{\textbf{Scenario 2: (+, +, ind)}}\\
\hspace{1em}2a & Frank, θ = (1, 1, ind) & 0.031 (p = 0.031) & 0.020 (p = 0.031) & 0.003 (p = 0.850) & 0.002 (p = 0.846)\\
\hspace{1em}2b & Frank, θ = (3, 3, ind) & 0.201 (p < 0.001) & 0.134 (p < 0.001) & -0.004 (p = 0.753) & -0.003 (p = 0.752)\\
\hspace{1em}2c & Frank, θ = (5, 5, ind) & 0.400 (p < 0.001) & 0.269 (p < 0.001) & -0.028 (p = 0.049) & -0.019 (p = 0.047)\\
\hspace{1em}2d & Gumbel, θ = (1.5, 1.5, ind) & 0.270 (p < 0.001) & 0.183 (p < 0.001) & 0.025 (p = 0.081) & 0.017 (p = 0.079)\\
\hspace{1em}2e & Gumbel, θ = (2, 2, ind) & 0.467 (p < 0.001) & 0.325 (p < 0.001) & -0.014 (p = 0.308) & -0.010 (p = 0.314)\\
\hspace{1em}2f & Gumbel, θ = (3, 3, ind) & 0.734 (p < 0.001) & 0.545 (p < 0.001) & 0.019 (p = 0.176) & 0.013 (p = 0.178)\\
\hspace{1em}2g & Clayton, θ = (1, 1, ind) & 0.276 (p < 0.001) & 0.187 (p < 0.001) & 0.016 (p = 0.259) & 0.011 (p = 0.260)\\
\hspace{1em}2h & Clayton, θ = (3, 3, ind) & 0.645 (p < 0.001) & 0.468 (p < 0.001) & 0.008 (p = 0.581) & 0.005 (p = 0.580)\\
\hspace{1em}2i & Clayton, θ = (5, 5, ind) & 0.795 (p < 0.001) & 0.608 (p < 0.001) & -0.005 (p = 0.749) & -0.003 (p = 0.746)\\
\addlinespace[0.3em]
\multicolumn{6}{l}{\textbf{Scenario 3: (-, -, ind)}}\\
\hspace{1em}3a & Frank, θ = (-1, -1, ind) & 0.048 (p < 0.001) & 0.032 (p < 0.001) & 0.019 (p = 0.181) & 0.013 (p = 0.182)\\
\hspace{1em}3b & Frank, θ = (-3, -3, ind) & 0.182 (p < 0.001) & 0.122 (p < 0.001) & -0.010 (p = 0.474) & -0.007 (p = 0.477)\\
\hspace{1em}3c & Frank, θ = (-5, -5, ind) & 0.390 (p < 0.001) & 0.264 (p < 0.001) & -0.028 (p = 0.050) & -0.019 (p = 0.048)\\
\hspace{1em}3d & Gumbel (R), θ = (1.5, 1.5, ind) & 0.234 (p < 0.001) & 0.158 (p < 0.001) & -0.004 (p = 0.784) & -0.003 (p = 0.777)\\
\hspace{1em}3e & Gumbel (R), θ = (2, 2, ind) & 0.485 (p < 0.001) & 0.339 (p < 0.001) & -0.001 (p = 0.954) & -0.000 (p = 0.964)\\
\hspace{1em}3f & Gumbel (R), θ = (3, 3, ind) & 0.724 (p < 0.001) & 0.534 (p < 0.001) & -0.001 (p = 0.923) & -0.001 (p = 0.926)\\
\hspace{1em}3g & Clayton (R), θ = (1, 1, ind) & 0.259 (p < 0.001) & 0.175 (p < 0.001) & 0.005 (p = 0.701) & 0.004 (p = 0.689)\\
\hspace{1em}3h & Clayton (R), θ = (3, 3, ind) & 0.652 (p < 0.001) & 0.474 (p < 0.001) & 0.011 (p = 0.422) & 0.007 (p = 0.432)\\
\hspace{1em}3i & Clayton (R), θ = (5, 5, ind) & 0.796 (p < 0.001) & 0.609 (p < 0.001) & 0.004 (p = 0.760) & 0.003 (p = 0.761)\\
\addlinespace[0.3em]
\multicolumn{6}{l}{\textbf{Scenario 4: (+, -, ind)}}\\
\hspace{1em}4a & Frank, θ = (1, -1, ind) & -0.017 (p = 0.243) & -0.011 (p = 0.241) & 0.011 (p = 0.432) & 0.007 (p = 0.437)\\
\hspace{1em}4b & Frank, θ = (3, -3, ind) & -0.183 (p < 0.001) & -0.123 (p < 0.001) & 0.018 (p = 0.194) & 0.012 (p = 0.198)\\
\hspace{1em}4c & Frank, θ = (5, -5, ind) & -0.424 (p < 0.001) & -0.287 (p < 0.001) & -0.006 (p = 0.672) & -0.004 (p = 0.681)\\
\hspace{1em}4d & Gaussian, θ = (0.3, -0.3, ind) & -0.119 (p < 0.001) & -0.080 (p < 0.001) & -0.036 (p = 0.011) & -0.024 (p = 0.010)\\
\hspace{1em}4e & Gaussian, θ = (0.5, -0.5, ind) & -0.250 (p < 0.001) & -0.168 (p < 0.001) & 0.002 (p = 0.879) & 0.001 (p = 0.877)\\
\hspace{1em}4f & Gaussian, θ = (0.7, -0.7, ind) & -0.477 (p < 0.001) & -0.330 (p < 0.001) & -0.016 (p = 0.260) & -0.011 (p = 0.261)\\
\addlinespace[0.3em]
\multicolumn{6}{l}{\textbf{Scenario 5: (+, +, -)}}\\
\hspace{1em}5a & Frank, θ = (1, 1, -1) & -0.156 (p < 0.001) & -0.105 (p < 0.001) & -0.190 (p < 0.001) & -0.128 (p < 0.001)\\
\hspace{1em}5b & Frank, θ = (1, 1, -3) & -0.410 (p < 0.001) & -0.279 (p < 0.001) & -0.452 (p < 0.001) & -0.309 (p < 0.001)\\
\hspace{1em}5c & Frank, θ = (1, 1, -5) & -0.598 (p < 0.001) & -0.419 (p < 0.001) & -0.648 (p < 0.001) & -0.460 (p < 0.001)\\
\hspace{1em}5d & Gaussian, θ = (0.2, 0.2, -0.2) & -0.147 (p < 0.001) & -0.098 (p < 0.001) & -0.190 (p < 0.001) & -0.128 (p < 0.001)\\
\hspace{1em}5e & Gaussian, θ = (0.2, 0.2, -0.5) & -0.406 (p < 0.001) & -0.278 (p < 0.001) & -0.467 (p < 0.001) & -0.322 (p < 0.001)\\
\hspace{1em}5f & Gaussian, θ = (0.2, 0.2, -0.7) & -0.638 (p < 0.001) & -0.456 (p < 0.001) & -0.699 (p < 0.001) & -0.508 (p < 0.001)\\
\hspace{1em}5g & Frank, θ = (5, 5, -1) & 0.320 (p < 0.001) & 0.214 (p < 0.001) & -0.170 (p < 0.001) & -0.114 (p < 0.001)\\
\hspace{1em}5h & Frank, θ = (5, 5, -0.5) & 0.367 (p < 0.001) & 0.247 (p < 0.001) & -0.095 (p < 0.001) & -0.064 (p < 0.001)\\
\hspace{1em}5i & Gaussian, θ = (0.7, 0.7, -0.2) & 0.386 (p < 0.001) & 0.263 (p < 0.001) & -0.202 (p < 0.001) & -0.136 (p < 0.001)\\
\hspace{1em}5j & Gaussian, θ = (0.7, 0.7, -0.1) & 0.425 (p < 0.001) & 0.292 (p < 0.001) & -0.090 (p < 0.001) & -0.060 (p < 0.001)\\
\addlinespace[0.3em]
\multicolumn{6}{l}{\textbf{Scenario 6: (-, -, +)}}\\
\hspace{1em}6a & Frank, θ = (-1, -1, 1) & 0.186 (p < 0.001) & 0.124 (p < 0.001) & 0.166 (p < 0.001) & 0.110 (p < 0.001)\\
\hspace{1em}6b & Frank, θ = (-3, -3, 1) & 0.332 (p < 0.001) & 0.226 (p < 0.001) & 0.184 (p < 0.001) & 0.123 (p < 0.001)\\
\hspace{1em}6c & Frank, θ = (-5, -5, 1) & 0.487 (p < 0.001) & 0.336 (p < 0.001) & 0.142 (p < 0.001) & 0.095 (p < 0.001)\\
\hspace{1em}6d & Gumbel (R), θ = (1.5, 1.5, 1.5) & 0.614 (p < 0.001) & 0.441 (p < 0.001) & 0.485 (p < 0.001) & 0.340 (p < 0.001)\\
\hspace{1em}6e & Gumbel (R), θ = (2, 2, 1.5) & 0.728 (p < 0.001) & 0.541 (p < 0.001) & 0.465 (p < 0.001) & 0.325 (p < 0.001)\\
\hspace{1em}6f & Gumbel (R), θ = (3, 3, 1.5) & 0.872 (p < 0.001) & 0.692 (p < 0.001) & 0.484 (p < 0.001) & 0.340 (p < 0.001)\\
\hspace{1em}6g & Clayton (R), θ = (1, 1, 1) & 0.615 (p < 0.001) & 0.445 (p < 0.001) & 0.479 (p < 0.001) & 0.335 (p < 0.001)\\
\hspace{1em}6h & Clayton (R), θ = (3, 3, 1) & 0.813 (p < 0.001) & 0.630 (p < 0.001) & 0.486 (p < 0.001) & 0.339 (p < 0.001)\\
\hspace{1em}6i & Clayton (R), θ = (5, 5, 1) & 0.896 (p < 0.001) & 0.730 (p < 0.001) & 0.473 (p < 0.001) & 0.329 (p < 0.001)\\
\addlinespace[0.3em]
\multicolumn{6}{l}{\textbf{Scenario 7: (-, -, -)}}\\
\hspace{1em}7a & Frank, θ = (-1, -1, -1) & -0.136 (p < 0.001) & -0.091 (p < 0.001) & -0.165 (p < 0.001) & -0.110 (p < 0.001)\\
\hspace{1em}7b & Frank, θ = (-1, -1, -3) & -0.417 (p < 0.001) & -0.285 (p < 0.001) & -0.459 (p < 0.001) & -0.316 (p < 0.001)\\
\hspace{1em}7c & Frank, θ = (-1, -1, -5) & -0.586 (p < 0.001) & -0.410 (p < 0.001) & -0.636 (p < 0.001) & -0.451 (p < 0.001)\\
\hspace{1em}7d & Gumbel (R), θ = (1.5, 1.5, 1.5) & -0.116 (p < 0.001) & -0.080 (p < 0.001) & -0.479 (p < 0.001) & -0.335 (p < 0.001)\\
\hspace{1em}7e & Gumbel (R), θ = (1.5, 1.5, 2) & -0.225 (p < 0.001) & -0.157 (p < 0.001) & -0.670 (p < 0.001) & -0.489 (p < 0.001)\\
\hspace{1em}7f & Gumbel (R), θ = (1.5, 1.5, 3) & -0.380 (p < 0.001) & -0.269 (p < 0.001) & -0.855 (p < 0.001) & -0.674 (p < 0.001)\\
\hspace{1em}7g & Clayton (R), θ = (1, 1, 1) & -0.050 (p < 0.001) & -0.038 (p < 0.001) & -0.467 (p < 0.001) & -0.325 (p < 0.001)\\
\hspace{1em}7h & Clayton (R), θ = (1, 1, 3) & -0.293 (p < 0.001) & -0.215 (p < 0.001) & -0.790 (p < 0.001) & -0.603 (p < 0.001)\\
\hspace{1em}7i & Clayton (R), θ = (1, 1, 5) & -0.343 (p < 0.001) & -0.257 (p < 0.001) & -0.883 (p < 0.001) & -0.712 (p < 0.001)\\
\addlinespace[0.3em]
\multicolumn{6}{l}{\textbf{Scenario 8: (+, -, +)}}\\
\hspace{1em}8a & Frank, θ = (1, -1, 1) & 0.119 (p < 0.001) & 0.079 (p < 0.001) & 0.147 (p < 0.001) & 0.098 (p < 0.001)\\
\hspace{1em}8b & Frank, θ = (3, -3, 1) & -0.096 (p < 0.001) & -0.063 (p < 0.001) & 0.144 (p < 0.001) & 0.096 (p < 0.001)\\
\hspace{1em}8c & Frank, θ = (5, -5, 1) & -0.322 (p < 0.001) & -0.214 (p < 0.001) & 0.159 (p < 0.001) & 0.106 (p < 0.001)\\
\hspace{1em}8d & Gaussian, θ = (0.2, -0.2, 0.2) & 0.134 (p < 0.001) & 0.090 (p < 0.001) & 0.183 (p < 0.001) & 0.123 (p < 0.001)\\
\hspace{1em}8e & Gaussian, θ = (0.5, -0.5, 0.2) & -0.091 (p < 0.001) & -0.061 (p < 0.001) & 0.192 (p < 0.001) & 0.129 (p < 0.001)\\
\hspace{1em}8f & Gaussian, θ = (0.7, -0.7, 0.2) & -0.383 (p < 0.001) & -0.261 (p < 0.001) & 0.186 (p < 0.001) & 0.124 (p < 0.001)\\
\addlinespace[0.3em]
\multicolumn{6}{l}{\textbf{Scenario 9: (+, -, -)}}\\
\hspace{1em}9a & Frank, θ = (1, -1, -1) & -0.191 (p < 0.001) & -0.128 (p < 0.001) & -0.171 (p < 0.001) & -0.114 (p < 0.001)\\
\hspace{1em}9b & Frank, θ = (3, -3, -1) & -0.360 (p < 0.001) & -0.243 (p < 0.001) & -0.187 (p < 0.001) & -0.124 (p < 0.001)\\
\hspace{1em}9c & Frank, θ = (5, -5, -1) & -0.497 (p < 0.001) & -0.342 (p < 0.001) & -0.156 (p < 0.001) & -0.104 (p < 0.001)\\
\hspace{1em}9d & Gaussian, θ = (0.2, -0.2, -0.2) & -0.224 (p < 0.001) & -0.151 (p < 0.001) & -0.197 (p < 0.001) & -0.133 (p < 0.001)\\
\hspace{1em}9e & Gaussian, θ = (0.5, -0.5, -0.2) & -0.387 (p < 0.001) & -0.264 (p < 0.001) & -0.186 (p < 0.001) & -0.124 (p < 0.001)\\
\hspace{1em}9f & Gaussian, θ = (0.7, -0.7, -0.2) & -0.574 (p < 0.001) & -0.404 (p < 0.001) & -0.197 (p < 0.001) & -0.132 (p < 0.001)\\
\addlinespace[0.3em]
\multicolumn{6}{l}{\textbf{Scenario 10: (ind, ind, $\cdot$)}}\\
\hspace{1em}10a & Frank, θ = (ind, ind, -1) & -0.179 (p < 0.001) & -0.120 (p < 0.001) & -0.179 (p < 0.001) & -0.120 (p < 0.001)\\
\hspace{1em}10b & Frank, θ = (ind, ind, 1) & 0.165 (p < 0.001) & 0.110 (p < 0.001) & 0.165 (p < 0.001) & 0.110 (p < 0.001)\\
\hspace{1em}10c & Indep, θ = (ind, ind, ind) & -0.019 (p = 0.169) & -0.013 (p = 0.169) & -0.019 (p = 0.169) & -0.013 (p = 0.169)\\
\addlinespace[0.3em]
\multicolumn{6}{l}{\textbf{Scenario 11: Simplifying assumption violated}}\\
\hspace{1em}11a & Frank, θ(z) = exp(z) & 0.272 (p < 0.001) & 0.183 (p < 0.001) & 0.259 (p < 0.001) & 0.174 (p < 0.001)\\
\hspace{1em}11b & Frank, θ(z) = -exp(z) & -0.253 (p < 0.001) & -0.170 (p < 0.001) & -0.289 (p < 0.001) & -0.195 (p < 0.001)\\
\hspace{1em}11c & Gaussian, θ(z) = 1-2z & 0.006 (p = 0.684) & 0.003 (p = 0.754) & -0.020 (p = 0.161) & -0.014 (p = 0.138)\\*
\end{longtable}
\endgroup{}

\end{document}